\documentclass[letterpaper,conference,10pt] {IEEEtran}
\usepackage{amsmath}
\usepackage{booktabs} 
\usepackage{amssymb,amsthm}
\usepackage{bm}
\usepackage{color}
\usepackage{graphicx}
\usepackage{empheq}
\usepackage[linesnumbered,ruled,lined]{algorithm2e}
\usepackage[outdir=./]{epstopdf}
\usepackage{caption}
\usepackage{lipsum} % for dummy text
\usepackage{lipsum}
\usepackage{cite}
\usepackage{diagbox}
\usepackage{multirow}
\usepackage{subcaption}
\usepackage{fancyhdr}
\SetKwInput{Input}{input}
\SetKwInput{Output}{output}

\newtheorem{lemma}{Lemma}
\usepackage{array}
\newcolumntype{L}[1]{>{\raggedright\let\newline\\\arraybackslash\hspace{0pt}}m{#1}}
\newcolumntype{C}[1]{>{\centering\let\newline\\\arraybackslash\hspace{0pt}}m{#1}}
\newcolumntype{R}[1]{>{\raggedleft\let\newline\\\arraybackslash\hspace{0pt}}m{#1}}

\theoremstyle{remark}

\IEEEoverridecommandlockouts
\def\specialpapernotice#1{\if@confmode%
	\def\@specialpapernotice{{\sublargesize\textit{#1}\vspace*{1em}}}%
	\else%
	\def\@specialpapernotice{{\\*[1.5ex]\sublargesize\textit{#1}}\vspace*{-2ex}}%
	\fi}

\newcommand{\T}{{\scriptscriptstyle\mathsf{T}}}
\renewcommand{\H}{{\scriptscriptstyle\mathsf{H}}}

\begin{document}
\setlength{\abovedisplayskip}{3.0pt}
\setlength{\belowdisplayskip}{3.0pt}
 
	\title{Beamforming Design for Max-Min Fairness Performance Balancing in ISAC Systems\vspace{-0.25cm}
   }

		\author{
			
		\IEEEauthorblockN{Tianyu Fang, Nhan Thanh Nguyen, and Markku Juntti}
		\IEEEauthorblockA{
		Centre for Wireless communications, University of Oulu, P.O.Box 4500, FI-90014, Finland\\
			Email: \{tianyu.fang, nhan.nguyen, markku.juntti\}@oulu.fi}\vspace{-0.75cm}		
			}

\maketitle

\thispagestyle{empty}
\pagestyle{empty}
\begin{abstract}
Integrated sensing and communications (ISAC) is envisioned as a key technology for future wireless communications. In this paper, we consider a downlink monostatic ISAC system wherein the base station serves multiple communications users and sensing targets at the same time in the presence of clutter. We aim at both guaranteeing fairness among the communications users while simultaneously balancing the performances of communications and sensing functionalities. Therefore, we optimize the transmit and receive beamformers to maximize the weighted minimum signal-to-interference and clutter-plus-noise ratios. The design problem is highly challenging due to the non-smooth and non-convex objective function and strongly coupled variables. We propose two efficient methods to solve the problem. First, we rely on fractional programming and transform the original problem into convex sub-problems, which can be solved with standard convex optimization tools. To further reduce the complexity and dependence on numerical tools, we develop a novel approach to address the inherent non-smoothness of the formulated problem. Finally, the efficiencies of the proposed designs are demonstrated by numerical results.% It is observed that while both the proposed schemes can achieve fairness in ISAC systems, the latter method has much low run time and complexity.  
\end{abstract}
\begin{IEEEkeywords}
 Max min fairness, MIMO systems,  integrated sensing and communications, non-smooth optimization.
\end{IEEEkeywords}

\section{Introduction}

%Wireless communications and radar sensing have traditionally thrived as distinct disciplines, each with its own objectives. However, the evolution of wireless networks has given rise to numerous applications, such as vehicular networks and smart houses, which necessitate both communications and sensing capabilities. This convergence has spurred the development of integrated sensing and communications (ISAC) for future sixth-generation (6G) wireless networks \cite{ma2020joint,zhang2021overview}. Recognized for its potential to enhance spectral efficiency and to reduce hardware costs, ISAC emerges as a promising technique for achieving spectrum sharing between communications and radar service on a unified platform.

Wireless communications and radar sensing are currently covered under the framework of integrated sensing and communications (ISAC)  \cite{ma2020joint,zhang2021overview}. Recognized for its potential to enhance spectral efficiency and to reduce hardware costs, ISAC emerges as a promising technique for achieving spectrum sharing between communications and radar service on a unified platform. ISAC systems can be classified based on their primary focus into 1) radar-centric, 2) communications-centric, and 3) joint designs. In the radar-centric approach, existing radar technologies are enhanced to include communications capabilities, achieved through methods like integrating digital messages~\cite{huang2020majorcom} into radar waveforms or modulating radar side lobes~\cite{Hassanien2016Dual}. Conversely, communications-centric approaches utilize conventional communications signals for environmental probing~\cite{Kumari2018IEEE80211ad}. In joint ISAC designs, transceivers are optimized to balance communications and sensing capabilities \cite{liu2018mu, nguyen2023jointssp, nguyen2023multiuser, nguyen2023joint}.

% The main focus of the early studies on ISAC lay on interference management approaches to enable the coexistence of communications and radar systems ~\cite{chiriyath2017radar, zheng2019radar, li2016optimum, martone2019joint}. 
The joint transceiver beamforming design for improving spectral efficiency constitutes a basic problem for ISAC systems~\cite{liu2018mu, liu2020joint, johnston2022mimo,choi2024joint,wang2024joint, zhang2023isac, wang2023optimizing,liu2021cramer, song2023intelligent, zhu2023cramer}. The achievable (sum) rate or signal-to-interference-plus-noise ratio (SINR) for each user are typical performance metrics for the communications functionality. For the sensing functionality, on the other hand, various metrics can be used depending on the sensing problem and setup. Those include the beampattern~\cite{liu2018mu, liu2020joint, johnston2022mimo}, signal-to-clutter-plus-noise ratio (SCNR)~\cite{choi2024joint,wang2024joint, zhang2023isac, wang2023optimizing}, Cramér--Rao lower bound of the target position estimate~\cite{liu2021cramer, song2023intelligent, zhu2023cramer}, and more~\cite{johnston2022mimo}. Specifically, in~\cite{liu2018mu, liu2020joint, johnston2022mimo}, transmit beamformers were designed to achieve desired sensing beampatterns under the SINR constraints. Furthermore, the SCNR alone \cite{choi2024joint, wang2024joint, zhang2023isac} or the weighted sum of SINR and SCNRs were employed \cite{wang2023optimizing} for beamforming design to guarantee fairness between communications users and sensing targets.
%While the metric can balance the detection probability or position estimation error for all targets, it has not been considered in the literature earlier. 
\par In this paper, we consider a multiuser multi-target monostatic ISAC system. To guarantee fairness among both communications users and simultaneously balance the resource usages with sensing, we aim at maximizing the weighted sum of the minimum SINR and the minimum SCNR, subject to per-antenna power constraint. The design problem is appealing due to its conceptual simplicity and the potential to consider both communications and sensing performances in the power domain without complicated mappings to higher-end performance metrics such as rate or target position estimation error. Regardless of its conceptual simplicity, the metric still results in a challenging optimization problem due to the non-smooth and non-convex objective function and strongly coupled variables. Therefore, we first rely on fractional programming (FP) to transform the original problem into convex sub-problems, which can be solved with convex optimization tools like CVX. However, one drawback of this approach is the high complexity and run time due to the dependence on CVX. We then propose an efficient and practical approach wherein the beamformer admits low-complexity closed-forms solutions. Finally, we present simulation results to evaluate the performance and run-time complexity of the two methods. It is observed that while both the proposed schemes can achieve fairness in ISAC systems, the latter method has a much lower run time and complexity. 

\par

\section{System Model and Problem Formulation}
\label{sec: system model}

\subsection{Signal Model}
\par We consider a downlink monostatic ISAC system, where the base station (BS) is equipped with $ L_{\mathrm{t}} $ transmit antennas and $ L_{\mathrm{r}} $ radar receive antennas. The BS simultaneously transmit probing signals to $K$ single-antenna communication users and $M$ point-like targets among $C$ clutter point-like targets. Let $ \mathcal K=\{1,\ldots,K\} $, $ \mathcal M=\{1,\ldots,M\} $, and $ \mathcal C=\{M+1,\ldots, M+C\} $ respectively denote the set of users, targets, and clutter targets. Let $ \mathbf s=[s_1,\ldots, s_K]^\T\in\mathbb C^{K\times 1} $ be the transmitted symbol vector and $ \mathbf W=[\mathbf w_1,\ldots,\mathbf w_K]\in\mathbb C^{L_{\mathrm{t}}\times K} $ be the corresponding precoding matrix at the BS. The transmitted signal at BS can be expressed as
$ \mathbf x= \mathbf W \mathbf s$, where it is assumed that $ \mathbf s\sim \mathcal{CN}(\mathbf 0,\mathbf I) $. 

\subsubsection{Communications Model}
The signal received by communications user $ k $ is given by
\begin{equation}\label{key}
	y_k=\mathbf h_k^\H\mathbf w_ks_k+\mathbf h^\H_k\sum_{j\neq k}^K\mathbf w_js_j +n_k, 
\end{equation}
where $ \mathbf h_k\in\mathbb C^{L_{\mathrm{t}}\times 1} $ is the channel vector between the BS and user $ k $, and $ n_k\sim\mathcal{CN}(0,\sigma_{\mathrm{c} k}^2)  $ is additive white Gaussian noise. Accordingly, the received signal-to-interference-plus-noise (SINR) of the intended symbol at user $ k $ is given as
\begin{equation}\label{key}
	\gamma_{\mathrm{c} k}=\frac{|\mathbf h_k^\H\mathbf w_k|^2}{\sum_{j\neq k}^K|\mathbf h^\H_k\mathbf w_j|^2+\sigma_{\mathrm{c} k}^2}.
\end{equation} 
The communications channels are modeled with Rician fading, given by $\mathbf h_k=\zeta_{\mathrm{c} k}\left(\sqrt{\frac{R_{\mathrm{F}}}{1+R_{\mathrm{F}}}}\mathbf h_k^{\text{LoS}}+\sqrt{\frac{1}{1+R_{\mathrm{F}}}}\mathbf h_k^{\text{NLoS}} \right)$, 
where $ \zeta_{\mathrm{c} k} $ denotes the large-scale path loss, $ R_{\mathrm{F}} $ refers to the Rician factor, while $ \mathbf h_k^{\text{LoS}} $ and $ \mathbf h_k^{\text{NLoS}}\sim\mathcal{CN}(\mathbf 0,\mathbf I) $ represent the small-scale line-of-sight (LoS) and non-line-of-sight (NLoS) components, respectively. The LoS component $ \mathbf h_k^{\text{LoS}} $ is expressed as $ \mathbf h_k^{\text{LoS}}=\alpha_{\mathrm{c} k} \mathbf a_{\mathrm{t}}(\phi_k) $, where $ \alpha_k $ is the complex gain, $ \phi_k $ refers to the angle-of-departure (AoD) and $ a_{\mathrm{t}}(\cdot) $ denotes the normalized transmit array steering vector at the BS.

\subsubsection{Sensing Model}
%For the sensing function, we focus on improving target detection performance and suppressing the clutter sources. To be specific, the received echo signal  at the BS is written as
The received echo signal  at the BS is 
\begin{equation}\label{key}
	\mathbf y_{\mathrm{s}}=\sum_{m=1}^{M}\mathbf G_m\mathbf x+\sum_{j=M+1}^{M+C}\mathbf G_j\mathbf x+\mathbf n_{\mathrm{s}},
\end{equation} 
where $ \mathbf G_i=\zeta_{\mathrm{s} i}\alpha_{\mathrm{s} i}\mathbf a_{\mathrm{r}}(\varphi_i)\mathbf a_{\mathrm{t}}^\H(\varphi_i), i=\{1,\ldots,M+C\} $, and $ \mathbf n_{\mathrm{s}}\sim \mathcal{CN}(0,\sigma_{\mathrm{s}}^2\mathbf I)  $ is the noise vector at the BS. Here, $ \zeta_{\mathrm{s} i}$, $\alpha_{\mathrm{s} i}$, and $\varphi_i $ represent the path loss, complex gain, and AoD, while $ \mathbf a_{\mathrm{r}}(\cdot)$ and $\mathbf a_{\mathrm{t}}(\cdot) $ are the normalized receive and transmit array steering vectors at the BS. 

\par Let $ \mathbf F=[\mathbf f_1,\ldots,\mathbf f_M]\in\mathbb C^{L_{\mathrm{r}}\times 1} $ be the receive combining matrix.  The combined signal for target $ m $ is given by
\begin{equation}\label{key}
	 r_m=\mathbf f_m^\H\mathbf y_{\mathrm{s}}= \mathbf f_m^\H\mathbf G_m\mathbf x+\mathbf f_m^\H\sum_{j\neq m}^{M+C}\mathbf G_j\mathbf x+\mathbf f_m^\H\mathbf n_{\mathrm{s}}.
\end{equation}
The SCNR for target $ m $ can be defined as
\begin{equation}\label{key}
	\gamma_{\mathrm{s} m}=\frac{\|\mathbf f_m^\H\mathbf G_m\mathbf W\|^2 }{\sum_{j\neq m}^{M+C}\| \mathbf f_m^\H\mathbf G_j\mathbf W\|^2+L_{\mathrm{r}}\sigma_{\mathrm{s}}^2\|\mathbf f_m\|^2 }.
\end{equation}

\subsection{Problem Formulation}
Our objective is to jointly optimize the transmit precoding matrix $ \mathbf W $ and receive combining matrix $ \mathbf F $ to achieve a balanced tradeoff between communications and sensing performance. Specifically, we seek a design that guarantees fairness among communications users and sensing targets. Therefore, we employ the minimal SINR and SCNR for the communications and sensing performance metrics, respectively. The joint design problem is formulated as
 \begin{subequations}
 	\label{P1}
 	\begin{align}
 		\max_{ \mathbf W,\mathbf F}&\,\, \min_{k\in\mathcal K}\{ \gamma_{\mathrm{c} k}\}+\delta_{\mathrm{d}} \min_{m\in\mathcal M}\{ \gamma_{\mathrm{s} m} \} \\
 	\label{PAPC}	\text{s.t.}\,\,
 		&\,\, \mathrm{diag}(\mathbf W\mathbf W^\H)\preceq P_{\mathrm{t}}\mathbf 1_{L_{\mathrm{t}}} / L_{\mathrm{t}}, 	\end{align}
 \end{subequations}
where $ P_{\mathrm{t}} $ is the transmit power budget, $\mathbf{1}_{L{\mathrm{t}}}$ represents a vector of all ones with a size of $L_{\mathrm{t}} \times 1$, and $ \delta_{\mathrm{d}} $ is the weight to control the tradeoff between communications and sensing performance. The per-antenna power constraint \eqref{PAPC} accounts for the individual power limits of the power amplifiers associated with each transmit antenna. Problem \eqref{P1} presents challenges due to the presence of non-smooth point-wise minimal functions, non-convex fractional SINRs and SCNRs, as well as strong coupling among variables. Next, we propose two efficient solutions for \eqref{P1}.

\section{Proposed Optimization Frameworks}
\label{sec: problem formulation}

\subsection{Fractional Programming Approach}\label{standard FP}
In this subsection, we propose an interior-point method (IPM)-based FP algorithm to solve problem \eqref{P1}. Specifically, the multiple non-convex fractional SINRs and SCNRs are decoupled into block-wise convex terms, and then we adopt the alternative optimization (AO) framework to solve each block of variables alternately.

First, we apply the FP technique named quadratic transform proposed in \cite{Shen2018fractional} to transform problem \eqref{P1} into a more tractable equivalent form. By introducing auxiliary variables $ \beta_{\mathrm{c} k}\in\mathbb C $ and $ \bm\beta_{\mathrm{s} m}\in \mathbb C^{1\times K }$, SINRs and SCNRs can be re-expressed as
\begin{equation}
	\begin{aligned}
		\gamma_{\mathrm{c} k}=&\max_{\beta_{\mathrm{c} k}}\,\, 2\Re\{\mathbf h_k^\H\mathbf w_k\beta_{\mathrm{c} k}^\H\}-|\beta_{\mathrm{c} k}|^2\left( \sum_{j\neq k}^K|\mathbf h^\H_k\mathbf w_j|^2+\sigma_{\mathrm{c} k}^2 \right),\\
		\gamma_{\mathrm{s} m}=&\max_{\bm\beta_{\mathrm{s} m}}\,\, 2\Re\{ \mathbf f_m^\H\mathbf G_m\mathbf W \bm\beta_{\mathrm{s} m}^\H\}\notag \\
		&\hspace{-0.25cm} - \|\bm\beta_{\mathrm{s} m}\|^2\left( \sum_{j\neq m}^{M+C}\| \mathbf f_m^\H\mathbf G_j\mathbf W\|^2+L_{\mathrm{r}}\sigma_{\mathrm{s}}^2\|\mathbf f_m\|^2 \right),
	\end{aligned} 
\end{equation}
respectively. Although problem \eqref{P1} is still non-convex with these new forms of the SINR and SCNR, an AO framework can be applied to decompose it into three convex sub-problems. For each block, each of $\mathbf W,\mathbf F $ and $ \bm \beta=\{\beta_{\mathrm{c} k},\forall k\in\mathcal K, \bm\beta_{\mathrm{s} m},\forall m\in\mathcal M\} $ is optimized by fixing the others. Specifically, the sub-problem of transmit beamforming matrix $ \mathbf W $ is convex and therefore can be globally solved by the CVX solver, while obtaining the optimal $\mathbf F $ and $\bm\beta$ is straightforward by examing their first-order optimal conditions. We refer to this method as standard FP and omit the detailed operations due to the space constraint. 

\subsection{Proposed Low-Complexity Design}
\subsubsection{Problem Reformulation}
\label{sec: algorithm}
Although \eqref{P1} can be solved by the above FP framework effectively, numerical tools are called in each iteration incurring a high computational complexity and run time, hindering its practical employment. To overcome this limitation, we develop a novel first-order low-complexity algorithm. At first, we reformulate \eqref{P1} as
\begin{equation}\label{P2}
 \max_{ \mathbf W\in\mathcal S,\mathbf F}\,\, \min_{k\in\mathcal K}\{ \log(1+\gamma_{\mathrm{c} k})\}+\delta_{\mathrm{l}} \min_{m\in\mathcal M}\{ \log(1+\gamma_{\mathrm{s} m} )\},
\end{equation}
where $  \mathcal S=\{\mathbf W: \mathrm{diag}(\mathbf W\mathbf W^\H)\preceq P_t\mathbf 1_{L_t}/L_t\} $, and $ \delta_{\mathrm{l}} $ is the new weight coefficient. Note that compared to the original objective function in \eqref{P1}, we have employed  $\log(1+\gamma_{\mathrm{c} k})$ and $\log(1+\gamma_{\mathrm{s} m} )$ for the communications and sensing metrics, respectively, in \eqref{P2}. As monotonically increasing functions, the use of these logarithmic functions does not affect the fairness of communications and sensing functions. By introducing two probability variables $ \mathbf z_{\mathrm{c}}=[z_{\mathrm{c}1},\ldots,z_{\mathrm{c} K}]^\T $ and $ \mathbf z_{\mathrm{s}}=[z_{\mathrm{s}1},\ldots,z_{\mathrm{s} M}]^\T $, we can rewrite problem \eqref{P2} as
\begin{equation}
    \begin{aligned}\label{P2t1}
		\max_{ \mathbf W\in\mathcal S,\mathbf F}\,\, \min_{\mathbf z_c\in \mathcal Z_c,\mathbf z_{\mathrm{s}}\in \mathcal Z_{\mathrm{s}} } &\sum_{k=1}^{K} z_{\mathrm{c} k}r_{\mathrm{c} k}+\delta_{\mathrm{l}} \sum_{m=1}^{M} z_{\mathrm{s} m}r_{\mathrm{s} m},
\end{aligned}
\end{equation}
where $ \mathcal Z_c=\{\mathbf z_c:\mathbf z_c\succeq\mathbf 0, \mathbf 1_K^\T\mathbf z_c=1 \} $ and $ \mathcal Z_{\mathrm{s}}=\{\mathbf z_{\mathrm{s}}:\mathbf z_{\mathrm{s}}\succeq\mathbf 0, \mathbf 1_M^\T\mathbf z_{\mathrm{s}}=1 \} $ are two compact convex simplices, and $r_{\mathrm{c} k}=\log(1+\gamma_{\mathrm{c} k}), r_{\mathrm{s} m}=\log(1+\gamma_{\mathrm{s} m} )$. By eliminating the non-smooth nature of problem \eqref{P2}, we can adopt first-order optimization methods to solve problem \eqref{P2t1}. Specifically, we first apply the Lagrangian and quadratic transformation to decouple the non-convex objective function. By introducing auxiliary variables $ \bm\xi_{\mathrm{c}}=[\xi_{\mathrm{c}1},\ldots,\xi_{\mathrm{c} K}]^\T\in\mathbb{R}^{K\times 1},\bm\xi_{\mathrm{s}}=[\xi_{\mathrm{s}1},\ldots,\xi_{\mathrm{s} M}]^\T\in\mathbb{R}^{M\times 1} ,\bm\theta_c=[\theta_{\mathrm{c}1},\ldots,\theta_{\mathrm{c} K} ]^\T\in\mathbb{C}^{K\times 1},\bm\Theta_{\mathrm{s}}=[\bm\theta_{\mathrm{s}1}^\H,\ldots,\bm\theta_{\mathrm{s} M}^\H ]^\H\in\mathbb{R}^{M\times K} $, problem \eqref{P2t1} can be equivalently reformulated as
\begin{equation}\label{P2t2}
	\begin{aligned}
		\max_{ \mathbf W\in\mathcal S,\mathbf F,\bm\xi,\bm\Theta}\,\, \min_{\mathbf z_c\in\mathcal Z_c,\mathbf z_{\mathrm{s}} \in\mathcal Z_{\mathrm{s}}} \sum_{k=1}^{K} z_{\mathrm{c} k}f_{\mathrm{c} k}+\delta_{\mathrm{l}} \sum_{m=1}^{M} z_{\mathrm{s} m} f_{\mathrm{s} m},
	\end{aligned}
\end{equation}
where $ \bm\xi=\{ \bm\xi_c,\bm\xi_{\mathrm{s}}\}, \bm\Theta=\{\bm\theta_c,\bm\Theta_{\mathrm{s}}\}$, and
\begin{align*}
		f_{\mathrm{c} k}&=\log(1+\xi_{\mathrm{c} k})+2\sqrt{1+\xi_{\mathrm{c} k} }\Re\{\mathbf h_k^\H\mathbf w_k\theta_{\mathrm{c} k}^\H \}\\
		&-|\theta_{\mathrm{c} k}|^2\left(\sum_{j=1}^{K}|\mathbf h_k^\H\mathbf w_j|^2+\sigma_{\mathrm{c} k}^2  \right)-\xi_{\mathrm{c} k},\\
		f_{\mathrm{s} m}&=\log(1+\xi_{\mathrm{s} m})+2\sqrt{1+\xi_{\mathrm{s} m}}\Re\{\mathbf f_m^\H\mathbf G_m\mathbf W\bm\theta_{\mathrm{s} m}^\H \}\\
		&-\|\bm\theta_{\mathrm{s} m}\|^2\left(\sum_{j=1}^{M+C}\| \mathbf f_m^\H\mathbf G_j\mathbf W\|^2+L_{\mathrm{r}}\sigma_{\mathrm{s}}^2\|\mathbf f_m\|^2  \right)-\xi_{\mathrm{s} m}.
\end{align*}
Although problem \eqref{P2t2} is still non-convex when all variables are jointly optimized, it is convex over each variable when fixing the others. Therefore, we employ the AO method to solve problem \eqref{P2t2}. Specifically, in this method, each variable block is solved by fixing the others, as elaborated next.

\subsubsection{AO Method for Solving \eqref{P2t2}}

% Utilizing the AO framework, .
\paragraph{Update $ \mathbf z_c, \mathbf z_{\mathrm{s}} $}
Note that for a given primal variable $ \{\mathbf W,\mathbf F,\bm\xi,\bm\Theta\} $, the optimal solution to the inner minimization occurs when $ z_{\mathrm{c},i}=1, i\in\mathcal K $ and $ z_{\mathrm{s},j}=1, j\in\mathcal M $, indicating that user $ i $ and target $ j $ achieve the worst-case SINR and SCNR, respectively. However, such vertex point at each iteration would lead to oscillation when it lies the neighborhood of any locally optimal point of problem \eqref{P2}. Inspired by the well-known log-sum-exp smoothing function \cite{xu2001smoothing}, we propose to employ the following closed-form approximation for $ \mathbf z_c $ and $ \mathbf z_{\mathrm{s}} $ at each iteration:
\begin{subequations}\label{Update z}
 	\begin{align}
 		z_{\mathrm{c} k}&=\frac{\exp(-\mu f_{\mathrm{c} k} )}{\sum_{j=1}^{K}\exp(-\mu f_{\mathrm{c}j}) }, \forall k\in\mathcal K,\\
 		z_{\mathrm{s} m}&=\frac{\exp(-\mu f_{\mathrm{s} m} )}{\sum_{j=1}^{M}\exp(-\mu f_{\mathrm{s}j}) }, \forall m\in\mathcal M,
 	\end{align}
\end{subequations} 
where $ \mu $ is a smooth parameter. Note that the approximation becomes tight as $\mu\rightarrow \infty$.

\paragraph{Update $ \bm\xi,\bm\Theta,\mathbf F $}
Given other variables fixed, the subproblems for $ \bm\xi, \bm\Theta $ and $ \mathbf F $ are unconstrained convex problems, respectively. By examining their first-order optimal conditions, the optimal solutions for $ \bm\xi , \bm\Theta $ and $ \mathbf F $ are given by
\begin{subequations}\label{Update auxi}
	\begin{align}
		&\xi_{\mathrm{c} k}=\gamma_{\mathrm{c} k},\quad\quad\xi_{\mathrm{s} m}=\gamma_{\mathrm{s} m},\\
		&\theta_{\mathrm{c} k}=\frac{\sqrt{1+\xi_{\mathrm{c} k}}\mathbf h_k^\H\mathbf w_k}{\sum_{j=1}^{K}|\mathbf h_k^\H\mathbf w_j|^2+\sigma_{\mathrm{c} k}^2 },\\
		&\bm\theta_{\mathrm{s} m}=\frac{\sqrt{1+\xi_{\mathrm{s} m} }\mathbf f_m^\H\mathbf G_m\mathbf W}{\sum_{j=1}^{M+C}\| \mathbf f_m^\H\mathbf G_j\mathbf W\|^2+L_{\mathrm{r}}\sigma_{\mathrm{s}}^2\|\mathbf f_m\|^2 },\\
		&\mathbf f_m=\! \frac{\sqrt{1+\xi_{\mathrm{s} m}}}{\|\bm\theta_{\mathrm{s} m}\|^2}\!\left(\sum_{j=1}^{M+C}\!\!\mathbf G_j\mathbf W\mathbf W^\H\mathbf G_j^\H +L_{\mathrm{r}}\sigma_{\mathrm{s}}^2\mathbf I \right)^{-1}\!\!\!\!\!\!\!\! \mathbf G_m\mathbf W\bm\theta_{\mathrm{s} m}^\H. 
	\end{align}
\end{subequations}
\paragraph{Update $ \mathbf W $}
\par Given other variables fixed, the subproblem with respect to $ \mathbf W $ is formulated as
\begin{align}\label{SubW}
		\max_{ \mathbf W\in\mathcal S}\,\,\,\,&2\Re\{\mathrm{tr}(\mathbf W (\mathbf X+\bm\Sigma_1^\H\mathbf H^\H) ) \}-\mathrm{tr}(\mathbf W\mathbf W^\H(\mathbf Y+\mathbf H\bm\Sigma_2\mathbf H^\H) )
\end{align}
where
\begin{equation*}\label{key}
	\begin{aligned}
		&\mathbf\Sigma_1=\mathrm{diag}\{z_{\mathrm{c}1}\sqrt{1+\xi_{\mathrm{c}1}}\theta_{\mathrm{c}1},\ldots,z_{\mathrm{c} k}\sqrt{1+\xi_{\mathrm{c} k}}\theta_{\mathrm{c} k} \},\\	
		&\mathbf H=[\mathbf h_1,\ldots,\mathbf h_K ], \mathbf \Sigma_2=\mathrm{diag}\{z_{\mathrm{c}1}|\theta_{\mathrm{c}1}|^2,\ldots,z_{\mathrm{c} k}|\theta_{\mathrm{c} k}|^2 \}\\
&\mathbf 	X=\delta_{\mathrm{l}}\sum_{m=1}^Mz_{\mathrm{s} m}\sqrt{1+\xi_{\mathrm{s} m}}\bm \theta_{\mathrm{s} m}^\H\mathbf f_m^\H\mathbf G_m,\\
&\mathbf Y=\delta_{\mathrm{l}}\sum_{m=1}^M\sum_{j=1}^{M+C}z_{\mathrm{s} m}\|\bm\theta_{\mathrm{s} m}\|^2\mathbf G_j\mathbf f_m\mathbf f_m^\H\mathbf G_j. 
	\end{aligned}
\end{equation*}
To solve this problem, we first introduce the following lemma.
\begin{lemma}\label{SCA}
	For any given positive semi-definite Hermitian matrix $ \mathbf A \in\mathbb{C}^{L_{\mathrm{t}}\times L_{\mathrm{t}}}$, we have
	\begin{equation}\label{key}
		\mathrm{tr}(\mathbf W\mathbf W^\H\mathbf A)\geq 2\Re\{\mathbf P\mathbf W^\H\mathbf A \}-\mathrm{tr}(\mathbf P\mathbf P^\H\mathbf A ),
	\end{equation}
where $ \mathbf P\in\mathbb C^{L_{\mathrm{t}}\times K} $ is an auxiliary matrix, and the equality is achieved if and only if $ \mathbf W=\mathbf P $.
\end{lemma}
\begin{proof}
	The result in Lemma \ref{SCA} can be easily verified by the inequality $ \mathrm{tr}((\mathbf W-\mathbf P)(\mathbf W-\mathbf P)^\H\mathbf A  )\geq 0 $.
\end{proof}
Based on Lemma \ref{SCA}, we obtain the following linear approximation of problem \eqref{SubW} for a given $ \mathbf P $:
\begin{equation}\label{linear approximation}
	\max_{\mathbf W\in\mathcal S} \Re\{\mathrm{tr}(\mathbf W(\bm\Sigma_1^\H\mathbf H^\H+\mathbf X )+\mathbf P\mathbf W^\H( \lambda\mathbf I-\mathbf H\bm\Sigma_2\mathbf H^\H-\mathbf Y ) \} ,
\end{equation}
where $ \lambda $ is an constant such that $ \lambda\mathbf I-\mathbf H\bm\Sigma_2\mathbf H^\H-\mathbf Y  $ is a positive semi-definite matrix. It is suggested to set it as the dominant eigenvalue of $ \mathbf H\bm\Sigma_2\mathbf H^\H + \mathbf Y  $. The optimal solution of problem \eqref{linear approximation} is then given by
\begin{equation}\label{Update W}
	\mathbf W=\bm\Pi_{\mathbf S}\left(\mathbf X^\H+\mathbf H\mathbf \Sigma_1+(\lambda\mathbf I-\mathbf H\bm\Sigma_2\mathbf H^\H-\mathbf Y ) \mathbf P \right),
\end{equation}
where $ \bm\Pi_{\mathcal S}(\cdot) $ denotes the projection of point $ \mathbf S$ onto set $ \mathcal S $, i.e.,
\begin{equation}\label{eq_projection}
	\bm\Pi_{\mathcal S}(\mathbf S)\triangleq \sqrt{P_{\mathrm{t}}/L}\left(\mathbf I\odot(\mathbf S\mathbf S^\H)\right)^{-1}\mathbf S,
\end{equation}
with $ \odot $ denoting the Hadamard product.

\par For clarity, we summarized the proposed low-complexity solution in Algorithm \ref{al1}. It starts with non-zero feasible matrices $\mathbf W^{[0]}$ and $\mathbf F^{[0]}$. Through iterative steps, we compute the probability variables $\mathbf z_c$ and $ \mathbf z_s$, auxiliary variables $\bm\xi$ and $\bm\Theta$, as well as the receive and transmit beamforming matrices $\mathbf F $ and $\mathbf W$, until the objective value in \eqref{P2t2} converges.

\addtolength{\topmargin}{0.01in}
\setlength{\textfloatsep}{7pt}	
\begin{algorithm}[t!]
    \small
	\textbf{Initialize}: $n\leftarrow0$, $\mathbf{W}^{[n]}$, $\mathbf{F}^{[n]}$\;
	\Repeat{The objective value in \eqref{P2t2} converges.}{
		$n\leftarrow n+1$\;
		Update $ \mathbf z_c^{[n]}$ and $ \mathbf z_{\mathrm{s}}^{[n]} $ by \eqref{Update z} \;
		Update $ \bm\xi^{[n]},\bm\Theta^{[n]}$ and  $\mathbf F^{[n]} $ by \eqref{Update auxi}\; 
		\For{$t=1:I_2$}{$ \mathbf P=\mathbf W^{[n]} $\;
	 Update $ \mathbf W^{[n]} $ by \eqref{Update W}\;}
		
	}	
	\caption{Proposed Low-Complexity Algorithm}
	\label{al1}				
\end{algorithm}

\subsection{Convergence and Complexity Analysis}
\label{sec: complexity}
For the standard FP method to solve problem \eqref{P1} specified in Sec. \ref{standard FP}, each block is optimally solved, generating a non-decreasing sequence bounded by the power constraint. Therefore, the standard FP method is guaranteed to converge with a given tolerance. However, the iterative calls for an IPM in the CVX solver incur high computational complexity, expressed as $\mathcal O(I_1(ML_{\mathrm{r}}^3+[KL_{\mathrm{t}}]^{3.5}) )  $, where $ I_1 $ denotes the number of iterations, $ ML_{\mathrm{r}}^3$ and $ [KL_{\mathrm{t}}]^{3.5} $ are the complexities for updating $ \mathbf F $ and $ \mathbf W $, respectively.  

For the proposed Algorithm \ref{al1} to solve problem \eqref{P2}, it is currently challenging to determine the specific range of the smooth parameter $ \mu $ that ensures convergence. Although we lack a rigorous theoretical analysis, all of our numerical experiments have shown stable convergence. Similarly to the standard FP method, the computational complexity of Algorithm \ref{al1} is primarily dominated by the updating of receive and transmit beamforming matrix, given as $ \mathcal O( ML_{\mathrm{r}}^3 ) $ and $ \mathcal O(I_2KL_{\mathrm{t}}^2 ) $, where $ I_2 $ denotes the number of inner iterations. Consequently, the overall computational complexity of Algorithm \ref{al1} is $ \mathcal O( I_1(ML_{\mathrm{r}}^3 +I_2KL_{\mathrm{t}}^2 ) ) $, which is significantly lower than the standard FP method.

%\vspace{-2mm}
\section{Numerical Results}
\label{sec: numerical results}
\par In this section, we evaluate the performance of the proposed FP method and Algorithm \ref{al1}. We employ the CVX toolbox \cite{grant2008cvx} to solve \eqref{P1} with the interior-point method. We model the path loss as $ \zeta_{\mathrm{c} k}=\zeta_0 d_{\mathrm{c} k}^{\iota_{\mathrm{c}}} $ and $ \zeta_{\mathrm{s} m}=\zeta_0 d_{\mathrm{s} m}^{\iota_{\mathrm{s}}} $, where $ \zeta_0=-30 $~dB represents the reference path loss at the reference distance of $ d=1 $~m, $ d_{\mathrm{c} k} $ and $ d_{\mathrm{s} m} $ are the distances from the BS to communications user $ k $ and sensing target $ m $, respectively, and $ \iota_c=3$ and $ \iota_{\mathrm{s}}=2 $ are path loss exponents. We set $d_{ck}=100+20\eta_{ck}$ and $d_{sm}=10+2\eta_{sm}$, where $\eta_{ck},\eta_{sm}\sim\mathcal{N}(0,1)$. The Rician factor for communications channels is $ R_{\mathrm{F}}=3 $dB and the complex gains are given by $ \alpha_{\mathrm{c} k}\sim\mathcal{CN}(0,1)$ and $\alpha_{\mathrm{s} i}\sim\mathcal{CN}(0,1) $, respectively. The AoDs $ \phi_k $ and $ \varphi_i $ follows the uniform distribution $ \mathcal U(-2\pi/3,2\pi/3) $ \cite{wang2024joint}. We consider a basic setup where the BS is equipped with $ L_{\mathrm{t}}=16 $ transmit antennas and $ L_{\mathrm{r}}=16 $ receive antennas, serving $ K=4 $ communications users while detecting $ M=2 $ targets among $ C=2 $ sources of clutter. The noise power at each user and the radar receiver is set to $ \sigma_{\mathrm{c} k}^2=-120 $~dBm and $ \sigma_{\mathrm{s}}^2=-120 $~dBm, and the transmit power budget is set so that the SNR becomes $P_t/ \sigma_{\mathrm{c} k}^2= 20 $~dB. %Specifically, we normalize all the channel coefficients by the noise power so that the CVX can solve each convex subproblem successfully. 
All the results are averaged by 100 channel realizations.

\begin{figure*}
    \begin{center}
        \begin{minipage}{0.32\textwidth}
            \hspace{-0.4cm}\includegraphics[width=1.1\linewidth]{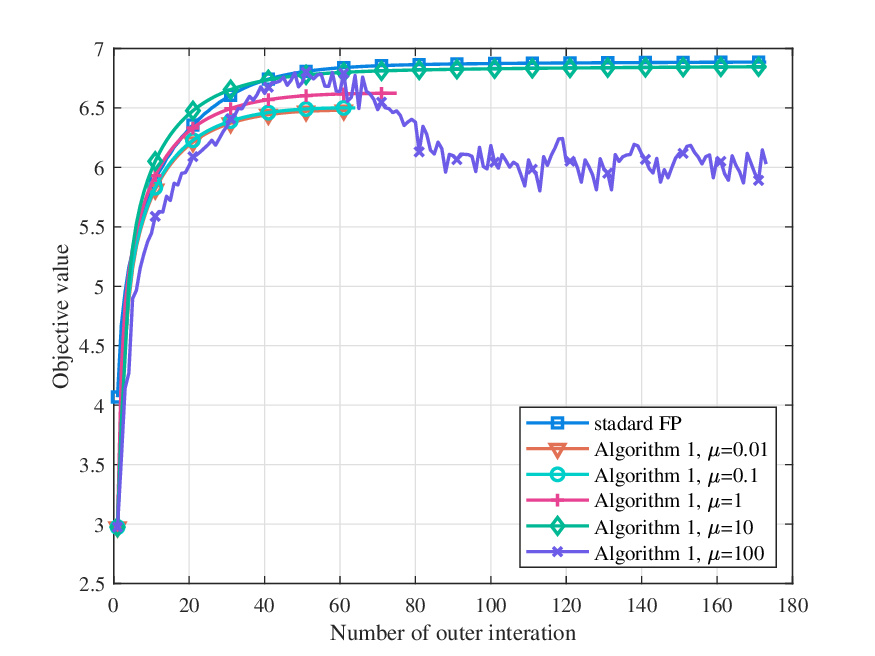}
            \caption{Convergence of the proposed FP scheme and Algorithm \ref{al1} with various $ \mu $.}
            \label{fig: convergence}
        \end{minipage}
        \hspace{0.005\textwidth} % Adjust horizontal space between minipages if necessary
        \begin{minipage}{0.32\textwidth}
            \hspace{-0.4cm}\includegraphics[width=1.1\linewidth]{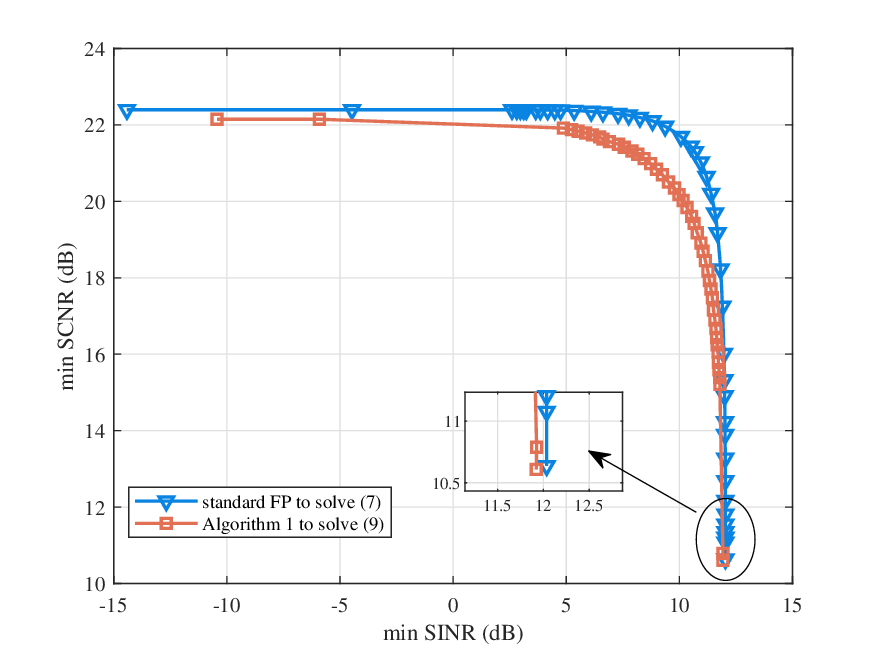}
            \caption{Tradeoff between the minimum SINR and SCNR of the  FP scheme and Algorithm \ref{al1}. }
            \label{fig: region}
        \end{minipage}
        \hspace{0.005\textwidth} % Adjust horizontal space between minipages if necessary
        \begin{minipage}{0.32\textwidth}
            \hspace{-0.4cm}\includegraphics[width=1.1\linewidth]{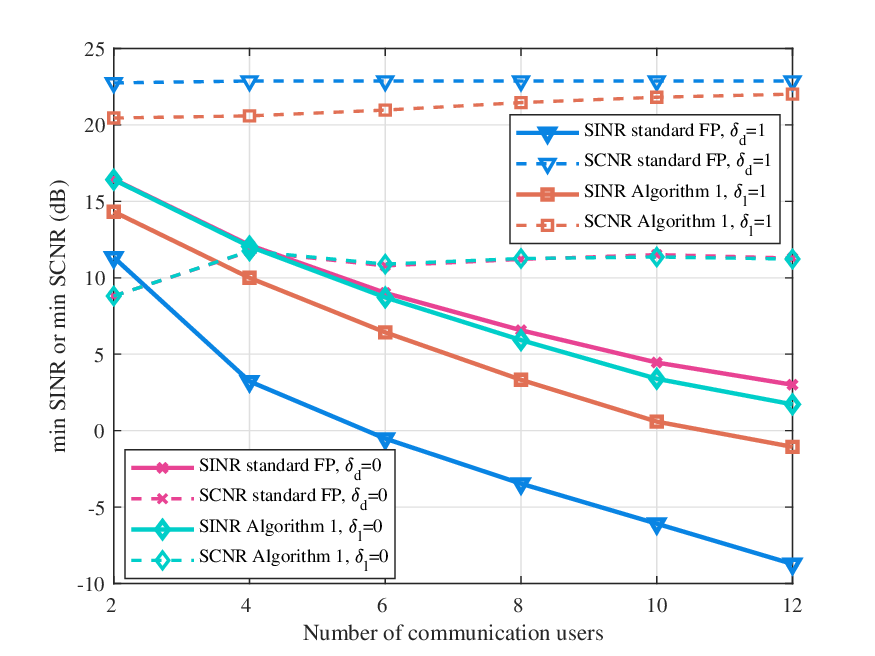}
            \caption{Performance comparison of the proposed FP scheme and Algorithm \ref{al1} for $K\in \{2,4,\ldots,12\}$.}
            \label{fig:2 region}
        \end{minipage}
    \end{center}
    \vspace{-0.5cm} % Adjust the amount of space here as per your requirement
\end{figure*}

  \begin{table}[t!]
		\begin{center}
			\caption{Run time (seconds) of the FP method and Algorithm \ref{al1} for $K \in \{2,4,\ldots,12 \}$ and $\delta_d=\delta_l \in \{0, 1\}$.}
			\label{tab:user}
			\resizebox{\linewidth}{!}{\begin{tabular}{ccccccc}
				\toprule
				Schemes & $K = 2$&$K = 4$&$K = 6$&$K = 8$&$K = 10$&$K = 12$
      \\
				\hline\hline
				standard FP, $\delta_d=0$ & 8.78&13.98&21.06&31.60&52.26&57.47 \\
    \hline
    			standard FP, $\delta_d=1$ & 157.26&405.19&395.30&403.51&333.49&284.45 \\
				\hline
				Algorithm \ref{al1}, $\delta_l=0$ &0.39& 0.56&0.40&0.26&0.14&0.13\\
    \hline
				Algorithm \ref{al1}, $\delta_l=1$ &1.14& 1.28&1.46&1.52&1.57&1.56\\
				\bottomrule\vspace{-0.75cm}
			\end{tabular}}
		\end{center}
	\end{table}
 
 \begin{table}[t!]
		\begin{center}
			\caption{Communications SINRs and sensing SCNRs (dB) achieved by the FP method and Algorithm \ref{al1}.}
			\label{tab:individual}
			\resizebox{\linewidth}{!}{\begin{tabular}{ccccccc}
				\toprule
				Schemes & user 1&user 2&user 3&user 4&target 1&target 2
      \\
				\hline\hline
				standard FP & 10.74&10.74&10.74&10.74&18.76&18.76 \\
				\hline
				Algorithm \ref{al1} &10.66& 10.71&11.15&10.53&21.57&18.31\\
				\bottomrule\vspace{-0.75cm}
			\end{tabular}}
		\end{center}
	\end{table}

Fig.\ \ref{fig: convergence} illustrates the convergence behavior of both the standard FP scheme and Algorithm \ref{al1} with $\delta_l=1$. Algorithm \ref{al1} showcases a tradeoff between stable convergence and performance, contingent upon the smooth parameter $ \mu $. Conversely, the standard FP method exhibits consistent convergence and superior performance, aligning with the prior analysis. However, the computational complexity of standard FP is much higher than that of Algorithm \ref{al1}, as analyzed in Section \ref{sec: complexity}. Furthermore, in the simulations of the FP method, we also observed an excessively long run time due to the employment of CVX, as shown in Table \ref{tab:user}. For example, with $K=4$, Algorithm \ref{al1} performs $25$ and $316$ times faster than the FP counterpart with $\delta_d = \delta_l= \{0, 1\}$, respectively.

Table \ref{tab:individual} presents the SINRs and SCNRs for all communications users and sensing targets obtained by the FP scheme and Algorithm \ref{al1} with $\delta_l=1$. Both schemes ensure good fairness among the users and targets. While Algorithm \ref{al1} exhibits a slight loss in minimum SINR compared to the FP scheme, its total SINR and SCNR are larger.% {\color{blue}This can be attributed to the reformulated objective functions, as shown in \eqref{PAPC} and \eqref{P2}, and the utilization of all power in Algorithm \ref{al1} via the projection in \eqref{eq_projection}.--Tianyu's view: This sentence may not reasonable. I haved checked the solution obtained from standard FP also satisfies the power constraint with equality. This solution achieves absolute fariness between communications users and sensing targets, which naturally contrast to the total SINR and SCNR. Actually the update formula of $\mathbf z_{ck},\mathbf z_{sm}$, i.e., \eqref{Update z},  achieves a tradeoff between sum rate and max min fairness rate. When $\mu=0$, $z_{ck}$ reduces to $\frac{1}{K}$ and thus the objective function reduces to sum rate. With the increase of $\mu$, the result more focus on the max min fairness rate. }

In Fig.\ \ref{fig: region}, we show the communications--sensing performance tradeoff with $\delta_d, \delta_l \in [0, 10^{6} ]$. With $\delta_d = \delta_l = 0$, the objective functions contain only the minimum communications SINR. In contrast, with $\delta_d, \delta_l \rightarrow \infty$, the objective functions are dominated by the minimum sensing SCNR. As expected, the standard FP method presents a better communications--sensing performance tradeoff because the transmit beamforming sub-problem can be globally solved by the CVX solver. However, its practicality in large-scale networks is hindered by the high computational and time complexities stemming from repeated utilization of the optimization toolbox, as justified in Section \ref{sec: complexity} and Table \ref{tab:user}.

%However, Algorithm \ref{al1} offers a better communications--sensing performance tradeoff. This is because the FP method solving \eqref{P1} directly tends to allocate more power to the function with better channel condition (in our simulation, sensing), while neglecting the other one, leading to unfairness between them. Conversely, taking logarithms of SINRs and SCNRs ensures that the function with weaker channel condition strength receives relatively more power compared to the former beamforming solution, as the derivative of the logarithmic function decreases quickly. 

Fig.\ \ref{fig:2 region} shows the communications and sensing performances versus the number of users $K$ with $\delta_d=\delta_l = \delta \in \{0, 1\}$. In both cases,  when $K$ increases, the minimum SINR decreases due to the more significant inter-user interference. In contrast, the minimum SCNR slightly increases because of the better chance that a sensing target is covered by more beams.  Furthermore, $\delta=0$ leads to better communications performance but degraded sensing performance compared to $\delta=1$. Indeed, with $\delta=0$, the objective functions in both \eqref{PAPC} and \eqref{P2} involve only the minimum SINR. With $\delta=1$, the FP method and Algorithm \ref{al1} achieve comparable sensing performance, while the latter yields significant improvements in communications performance due to the reformulated objective functions.% compared to the solution obtained forming soling \eqref{P1}. 

% Finally, in Table \ref{tab:user}, we present the run time required by the FP scheme and Algorithm \ref{al1} to achieve the results shown in Fig. \ref{fig:2 region}. It is evident that Algorithm \ref{al1}, with closed-form solutions, executes much faster than the FP scheme, which heavily relies on built-in computation toolboxes like CVX.

% {\color{red}Nhan's comment: We focus on the max-min fairness design, but the "fairness" aspects have not been highlighted in the numerical discussion.  Please consider highlighting that because currently only general "performance" is discussed. Furthermore, please enlarge the figures so they fully occupy the page width.}

\section{Conclusion}
\label{sec: conclusion}
\par We considered a monostatic ISAC system with multiple communications users, sensing targets, and clutter objects. Aiming at enhancing the fairness among communications users and balancing the performance with sensing functionality, we formulated the beamforming design problem to maximize the weighted sum of minimum communications SINRs and sensing SCNRs under the per-antenna power constraint. We then proposed two efficient methods, namely, the standard FP and low-complexity first-order algorithms. Both achieve a good tradeoff between communications and sensing performances. However, while the former exhibits a good convergence profile owing to leveraging CVX, it has high complexity and run time. In contrast, the first-order method leads to closed-form solutions for the transmit and receive beamformers, and thus, is more computational and time efficient than the standard FP counterpart. Simulation results show that both methods achieve a good tradeoff between communications and sensing performance. However, the latter approach exhibits significantly lower computational complexity and simulation time. In future work, we will elaborate on the actual communications rate and sensing accuracy performances in more detail.
 
\section*{Acknowledgement}
This research was supported by the Research Council of Finland through 6G Flagship (grant number: 346208) and project DIRECTION (grant number: 354901) and by CHIST-ERA via project PASSIONATE (grant number: 359817).

\bibliographystyle{IEEEtran}
\bibliography{IEEEabrv,reference}
\end{document}